\documentclass[a4paper,english,10pt]{article}

\usepackage[utf8]{inputenc}
\usepackage[a4paper]{geometry}
\usepackage{cite}
\usepackage{enumerate,hyperref,nameref}
\usepackage[inline]{enumitem}
\usepackage[algo2e,noline,tworuled]{algorithm2e}
\usepackage{amsmath,amsthm,amsfonts,amssymb,stmaryrd,mathtools,esvect}
\usepackage{tikz,graphicx}
\usetikzlibrary{backgrounds,calc,fit,shapes,positioning,arrows,intersections}

\theoremstyle{plain}\newtheorem{theorem}{Theorem}
\theoremstyle{plain}\newtheorem{corollary}[theorem]{Corollary}
\theoremstyle{plain}\newtheorem{proposition}[theorem]{Proposition}
\theoremstyle{plain}\newtheorem{lemma}[theorem]{Lemma}
\theoremstyle{plain}\newtheorem{definition}{Definition}
\theoremstyle{plain}\newtheorem{remark}{Remark}
\theoremstyle{plain}\newtheorem{example}[remark]{Example}

\makeatletter
\let\orgdescriptionlabel\descriptionlabel
\renewcommand*{\descriptionlabel}[1]{%
  \let\orglabel\label
  \let\label\@gobble
  \phantomsection
  \edef\@currentlabel{\ignorespaces #1\unskip}%
  \let\label\orglabel
  \orgdescriptionlabel{#1}%
}
\makeatother

\DeclareSymbolFont{symbolsC}{U}{txsyc}{m}{n}
\SetSymbolFont{symbolsC}{bold}{U}{txsyc}{bx}{n}
\DeclareFontSubstitution{U}{txsyc}{m}{n}
\DeclareMathSymbol{\multimapboth}{\mathrel}{symbolsC}{"13}

\newcommand{\?}[1]{}

\newcommand{\defeq}{\triangleq}
\newcommand{\defiff}{\stackrel{\triangle}{\iff}}
\newcommand{\mono}{\rightarrowtail}

\newcommand{\emb}{\mathbin{%
  \begin{tikzpicture}[baseline=-.6ex]\draw[right hook->] (0,0) -- (2ex,0);\end{tikzpicture}}}
\newcommand{\pemb}{\mathbin{%
  \begin{tikzpicture}[baseline=-.6ex]\draw[right hook-left to] (0,0) -- (2ex,0);\end{tikzpicture}}}
\newcommand{\pmap}{\rightharpoonup}
\newcommand{\face}[1]{\langle #1 \rangle}
\newcommand{\dom}{\textsf{\textit{dom}}}

\newcommand{\rng}{\textsf{\textit{img}}}
\newcommand{\ar}{\textsf{\textit{ar}}}
\newcommand{\prnt}{\textsf{\textit{prnt}}}
\newcommand{\ctrl}{\textsf{\textit{ctrl}}}
\newcommand{\link}{\textsf{\textit{link}}}
\newcommand{\src}{\mathsf{src}}

\newcommand{\mbb}[1]{\mathbb{#1}}

\newcommand{\msf}[1]{\mathsf{#1}}

\newcommand{\esf}[2]{#1^\mathsf{#2}}
\newcommand{\ephi}[1]{\esf{\phi}{#1}}

\newcommand\restr[2]{{\left.\kern-\nulldelimiterspace#1\vphantom{\big|}\right|_{#2}}}
\newcommand\corestr[2]{{\left.\kern-\nulldelimiterspace#1\vphantom{\big|}\right|^{#2}}}

\newcommand{\prt}{\mbb}
\newcommand{\eprt}[1]{\prt P^\mathsf{\kern1pt#1}}
\newcommand\embrestr[3][\?]{{\left.\kern-\nulldelimiterspace#3\vphantom{\big|}\right|_{#1,#2}}}

\newcommand{\adjto}{\mathrel{{\circ}\!{\to}}}

\hypersetup{
	colorlinks=true,
	linkcolor=black,
	citecolor=black,
	filecolor=black,
	urlcolor=black,
	pdftitle={Distributed execution of bigraphical reactive systems},
	pdfauthor={Alessio Mansutti and Marino Miculan and Marco  Peressotti},
	pdfsubject={Concurrency Theory},
	pdfkeywords={Concurrent and distributed graph transformations, Models of graph transformation, Formal graph languages, Multi-agent systems}
}

\title{Distributed execution of bigraphical reactive systems\thanks{
	This work is partially supported by MIUR PRIN project 2010LHT4KM, \emph{CINA}.}}

\author{
	\begin{tabular}{ccccc}
	Alessio Mansutti&\qquad& Marino Miculan&\qquad& Marco  Peressotti\\
	\small\href{mailto:alessio.mansutti@gmail.com}{\tt alessio.mansutti@gmail.com}
	&&
	\small\href{mailto:marino.miculan@uniud.it}{\tt marino.miculan@uniud.it}
	&&
	\small\href{mailto:marco.peressotti@uniud.it}{\tt marco.peressotti@uniud.it}
	\end{tabular}\\[5pt]
	\small	Laboratory of Models and Applications of Distributed Systems \\[-.8ex]
	\small	Department of Mathematics and Computer Science\\[-.8ex]
	\small	University of Udine, Italy\\
}
\date{}

\begin{document}
\maketitle

\begin{abstract}
	The \emph{bigraph embedding} problem is crucial for many
	results and tools about bigraphs and bigraphical reactive systems
	(BRS).  Current algorithms for computing bigraphical embeddings are
	\emph{centralized}, i.e. designed to run locally with a complete
	view of the guest and host bigraphs. In order to deal with large
	bigraphs, and to parallelize reactions, we present a
	\emph{decentralized algorithm}, which distributes both state and
	computation over several concurrent processes. This allows for
	distributed, parallel simulations where non-interfering reactions
	can be carried out concurrently; nevertheless, even in the worst
	case the complexity of this distributed algorithm is no worse than
	that of a centralized algorithm.
\end{abstract}

\section{Introduction}

\emph{Bigraphical Reactive Systems} (BRSs)
\cite{jm:popl03,milner:bigraphbook} are a flexible and expressive
meta-model for ubiquitous computation.  In the last decade, BRSs have
been successfully applied to the formalization of a wide range of
domain-specific calculi and models, from traditional programming
languages to process calculi for concurrency and mobility, from
business processes to systems biology; a non exhaustive list is
\cite{bdehn:fossacs06,bghhn:coord08,bgm:biobig,dhk:fcm,mp:br-tr13,mmp:dais14}.
Recently, BRSs have found a promising applications in
\emph{structure-aware agent-based computing:} the knowledge about the
(physical) world where the agents operate (e.g., drones, robots, etc.)
can be conveniently represented by means of BRSs
\cite{pkss:bigactors,sp:memo14}.
BRSs are appealing also because they provide a range of 
general results and tools, which can be readily instantiated with the
specific model under scrutiny: simulation tools, systematic
construction of compositional bisimulations \cite{jm:popl03},
graphical editors \cite{fph:gcm12}, general model checkers
\cite{pdh:sac12}, modular composition \cite{pdh:refine11}, stochastic
extensions \cite{kmt:mfps08}, etc.

This expressive power stems from the rich structure of
\emph{bigraphs}, which are the states of a bigraphic reactive system.
A bigraph is a compositional data structure describing at once both
the locations and the connections of (possibly nested) system
components.  To this end, bigraphs combine two independent graphical
structures over the same set of \emph{nodes}: a hierarchy of
\emph{places}, and a hypergraph of \emph{links}.  Intuitively, places
represent (physical) positions of agents, while links represent
logical connections between agents.  A simple example is shown in
Figure~\ref{fig:bigraph-comp}.

\begin{figure}[t]
	\centering 
	\includegraphics{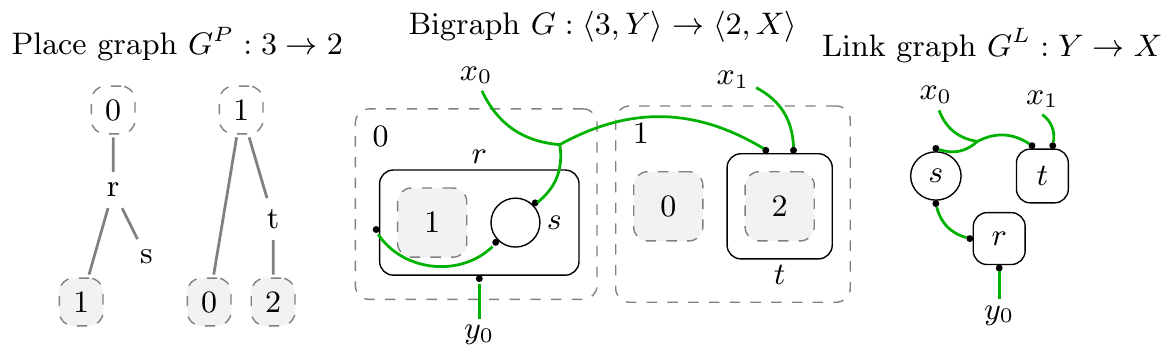}
	\caption{Forming a bigraph from a place graph and a link graph.}
	\label{fig:bigraph-comp}
\end{figure}

The behaviour of a BRS is defined by
a set of \emph{(parametric) reaction rules}, like in graph rewriting
\cite{graphtransformation}.  Applying a reaction rule to a bigraph
corresponds to find an embedding of the rule's \emph{redex} and
replace it with the corresponding \emph{reactum}. Thus, BRSs can be
run (or simulated) by the abstract machine depicted in
Figure~\ref{fig:brs-cycle} (or variants of it).  This machine is
composed by two main modules: the \emph{embedding engine} and the
\emph{reaction engine}.  The former keeps track of available redex
embeddings into the bigraph in the current machine state; the latter
is responsible of carrying out the reactions, in two steps:
\begin{enumerate*}[label=\em(\alph*)]
	\item 
		choosing an occurrence of a redex among those 
		provided by the embedding engine and
	\item 
		updating the machine state by performing the 
		chosen rewrite operation.
\end{enumerate*}

The choice of which reaction to perform is driven by
\emph{user-provided execution policies}.  A possible simple policy is
the random selection of any available reactions, while in
\cite{mmp:dais14} execution policies are based on agent believes,
intentions and goals.  Execution policies are outside the scope of
this paper, and we refer the reader to \cite{perrone:thesis} for other
examples. Here we mention LibBig, an extensible library for
bigraphical reactive systems (available at \url{http://mads.dimi.uniud.it/})
which offers easily customizable execution policies in the form of
cost-based embeddings where costs are defined at the component level
via \emph{attached properties}.

\begin{figure}[t]
	\centering 
	\includegraphics{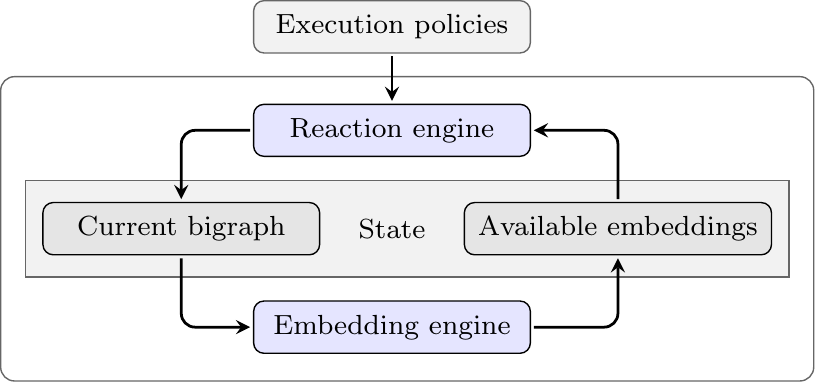}
	\caption{The open cycle of an abstract bigraphical machine.}
	\label{fig:brs-cycle}
\end{figure}

Therefore, computing bigraph embeddings (i.e., finding the occurrences
of a bigraph, called \emph{guest}, inside another one, called
\emph{host}) is a central issue in any implementation of a BRS
abstract machine.  The problem is known to be NP-complete
\cite{bmr:tgc14}, and some algorithms (or reductions) can be found in
the literature \cite{gdbh:implmatch,mp:memo14,sevegnani2010sat}.
However, existing algorithms assume a complete view of both the guest
and the host bigraphs.  This hinders the scalability of BRS execution
tools, especially on devices with low resources (like embedded
ones). Moreover, in a truly distributed setting (like in multi-agent
systems \cite{mmp:dais14}) the bigraph is scattered among many
machines; gathering it to a single ``knowledge manager'' in order to
calculate embeddings and apply the rewriting rules, would be
impractical.

In this paper, we aim to overcome these problems, by introducing an
algorithm for computing bigraphical embeddings in distributed settings
where bigraphs are spread across several cooperating processes. This
decentralized algorithm does not require a complete view of the host
bigraph, but retains the fundamental property of (eventually) computing every
possible embedding for the given host.  Thanks to the distributed
nature of the algorithm, this solution can scale to bigraphs that
cannot fit into the memory of a single process, hence too large to be
handled by existing implementations.  Moreover, the algorithm is
parallelized: several (non-interfering) reductions can be identified
and applied at once. In this paper we consider distributed host bigraphs only
since guest bigraphs are usually redexes of parametric reaction rules and
hence small enough to be handled even in presence
of scarce computational resources.

\begin{figure}[t]
	\centering
	\includegraphics{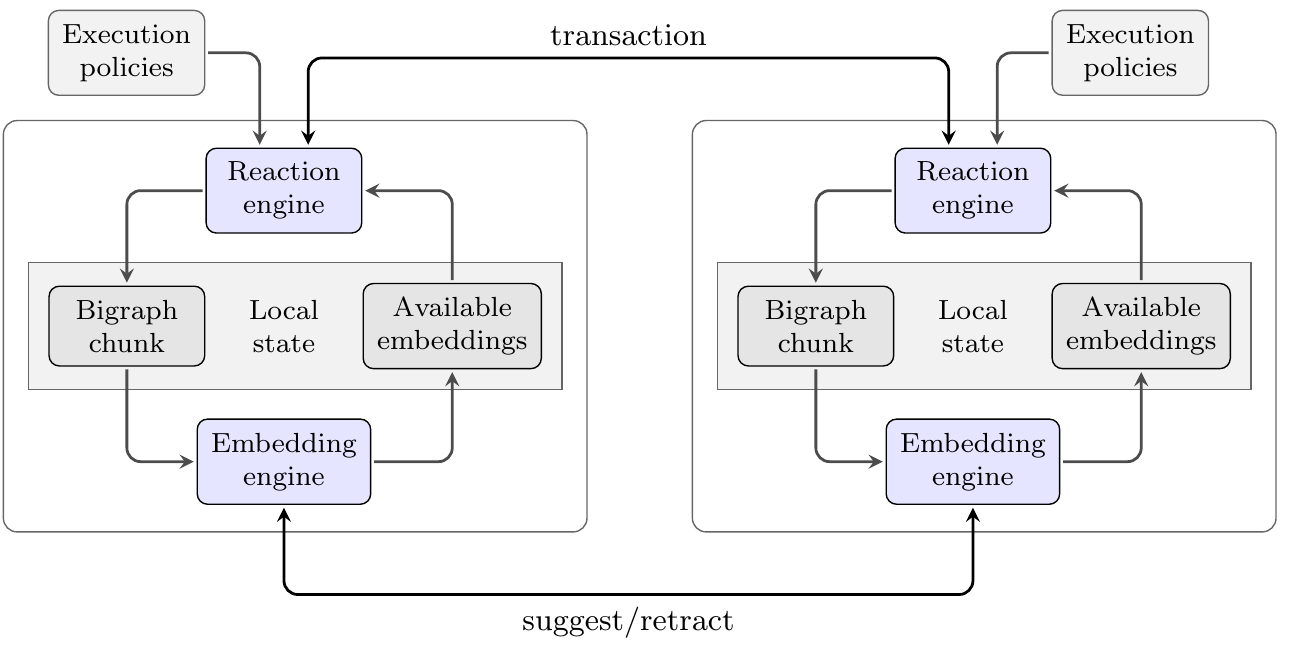}
	\caption{D-BAM: Distributed bigraphical abstract machine.}
	\label{fig:dbrs-arch}
\end{figure}

This algorithm is the core of a decentralized version of the abstract
bigraphical machine illustrated above. The architecture of this new
\emph{distributed bigraphical (abstract) machine} (D-BAM) is in
Figure~\ref{fig:dbrs-arch}.  Both computation and states are
distributed over a family of processes. Each process has only a
partial view of the global state and negotiates updates to its piece
of the global bigraph with its ``neighbouring processes''.  We assume
reliable asynchronous point-to-point communication between reliable
processes; this is a mild assumptions for a distributed system and can
be easily achieved e.g.~over unreliable channels.

This work extends and improves \cite{mpm:gcm14w} in several ways.
First, we introduce a new compact representation of partial
embeddings, reducing both network and memory footprint of the
distributed embedding algorithm; secondly, messages are routed across
the overlay network only to processes that can benefit from their
content (in \cite{mpm:gcm14w} messages were forwarded to the entire
neighbourhood). Moreover, we discuss some other heuristics and
partition strategies.

\paragraph{Synopsis}
In Section~\ref{sec:prelim} we briefly recall bigraphical reactive
systems and bigraph embeddings.  In Section~\ref{sec:pemb-cpemb} we
introduce the notion of \emph{partial bigraph embedding} and the
weaker notion of \emph{candidate} partial bigraph embedding.  In
Section~\ref{sec:dbrs} and Section~\ref{sec:demb} we describe the
D-BAM; in particular we show how to solve the embedding problem at its
core by means of a distributed algorithm which incrementally computes
(candidate) partial bigraph embeddings.  Conclusions and final remarks
are discussed in Section~\ref{sec:concl}.

\section{Bigraphs and their embeddings}
\label{sec:prelim}
In this section we briefly recall the notion of bigraphs, Bigraphical
Reactive Systems (BRS), and bigraph embedding; for more detail we
refer to \cite{milner:bigraphbook}.

\subsection{Bigraphical reactive systems}
\label{sec:brs}
The idea at the core of BRSs is that agents may interact in a
reconfigurable space, even if they are spatially separated.  This
means that two agents may be adjacent in two ways: they may be at the
same \emph{place}, or they may be connected by a \emph{link}.  Hence,
the state of the system is represented by a \emph{bigraph}, i.e., a
data structure combining two independent graphical structures over the
same set of \emph{nodes}: a hierarchy of \emph{places}, and a
hyper-graph of \emph{links}

\begin{definition}[Bigraph {\cite[Def.~2.3]{milner:bigraphbook}}]
	Let $\Sigma$ be a bigraphical signature (i.e.~a set of
	\emph{controls}, each associated with a finite arity).
	A \emph{bigraph} $G$ over $\Sigma$ is an object
	\[
		(V_G, E_G, \ctrl_G, \prnt_G, \link_G):\face{n_G,X_G}\to\face{m_G,Y_G}
	\]
	composed of two substructures (cf.~Figure~\ref{fig:bigraph-comp}): 
	a \emph{place graph}  $G^P=(V_G, \ctrl_G, \prnt_G):{n_G\to m_G}$ and a
	\emph{link graph} $G^L=(V_G,E_G,\ctrl_G,\link_G):{X_G\to Y_G}$.
	The set $V_G$ is a finite set of nodes and to each of them is assigned a
	control in $\Sigma$ by the \emph{control map} $\ctrl_G : V_G\to \Sigma$.
	The set $E_G$ is a finite set of names called \emph{edges}.
	These structures present an inner interface (composed by $m_G$ and 
	$X_G$) and an outer one ($n_G$, $Y_G$) along which can be composed with
	other of their kind as long as they do not share any
	node or edge. In particular, $X_G$ and $Y_G$ are finite sets
	of names and $m_G$ and $n_G$ are finite ordinals (that index \emph{sites} and \emph{roots} respectively).
	On the side of $G^P$, nodes, sites and roots are organized in
	a forest described by the \emph{parent map} $\prnt_G : V_G 
	\uplus m_G \to V_G \uplus n_G$.
	On the side of $G^L$, nodes, edges and names of the inner and outer
	interface forms a hyper-graph described by the \emph{link map} 
	$\link_G : P_G\uplus X_G \to E_G \uplus Y_G$ which is a function from $X_G$ and ports $P_G$ (i.e.~elements of 
	the finite ordinal associated to each node by its control) to 
	edges $E_G$ and outer names $Y_G$.
\end{definition}

The dynamic behaviour of a system is described in terms of
\emph{reactions} of the form $a \rightarrowtriangle a'$ where $a,a'$
are agents, i.e.~bigraphs with inner interface $\face{0,\emptyset}$.
Reactions are defined by means of graph rewrite rules, which are pairs
of bigraphs $(R_L, R_R)$ equipped with a function $\eta$ from the 
sites of $R_R$ to those of $R_L$ called \emph{instantiation rule}.
A bigraphical encoding for the open reaction rule of the Ambient 
Calculus is shown in Figure~\ref{fig:amb-open} where redex and 
reactum are the bigraph on the left and the one on the right respectively
and the instantiation rule is drawn in red. A rule fires when its redex 
can be embedded into the agent; then, the matched part is replaced by 
the reactum and the parameters (i.e.~the substructures determined by 
the redex sites) are instantiated accordingly with $\eta$.

\begin{figure}[t]
  \centering
  \includegraphics{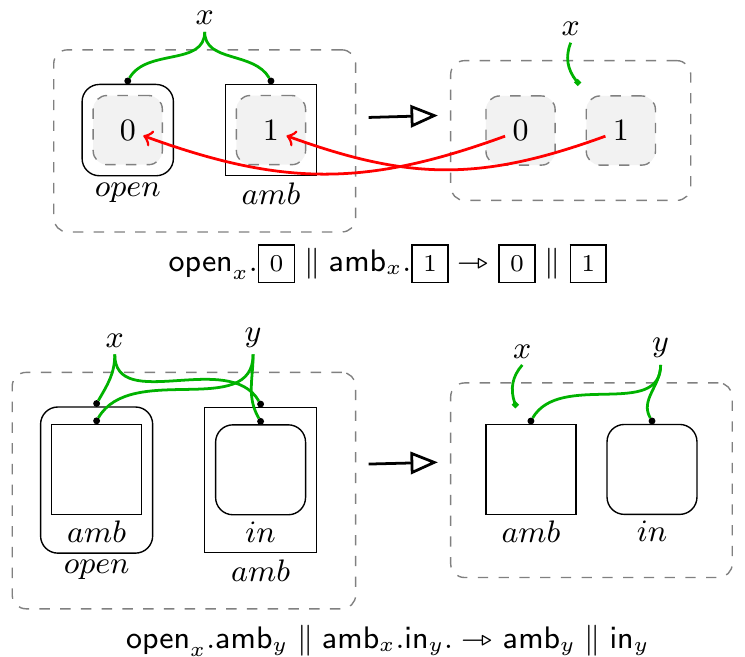}
  \caption{The open reaction rule of the Ambient Calculus (top)
  and an induced reaction.}
  \label{fig:amb-open}
\end{figure}

\subsection{Bigraph embeddings}
\label{sec:emb}
The following definitions are mainly taken from \cite{hoesgaard:thesis},
with minor modification to simplify the presentation of the 
distributed embedding algorithm (cf.~Section~\ref{sec:demb}).
As usual, we will exploit the orthogonality of the link and place graphs, 
by defining \emph{link and place graph embeddings} separately and then 
combine them to extend the notion to bigraphs. 

\vspace{-0.5ex}

\paragraph{Link graph}
Intuitively an embedding of link graphs is a structure preserving map
from one link graph (the \emph{guest}) to another (the \emph{host}). 
As one would expect from a graph 
embedding, this map contains a pair of injections: one for the nodes
and one for the edges (i.e., a support translation). The remaining
of the embedding map specifies how names of the inner and outer 
interfaces should be mapped into the host link graph. Outer names can
be mapped to any link; here injectivity is not required since a 
context can alias outer names. Dually, inner names can mapped to 
hyper-edges linking sets of points in the host link graph and such 
that every point is contained in at most one of these sets.

\begin{definition}[Link graph embedding {\cite[Def~7.5.1]{hoesgaard:thesis}}]\label{def:lge}
	Let $G : X_G \to Y_G$ and $H : X_H \to Y_H$ be two concrete link graphs. 
	A \emph{link graph embedding} $\phi : G \emb H$ is a map
	$\phi \defeq \ephi v \uplus \ephi e \uplus \ephi i \uplus \ephi o$
	(assigning nodes, edges, inner and outer names respectively)
	subject to the following conditions:
	\begin{description}\itemsep=0pt
	\item[(LGE-1)\label{def:lge-1}]
		$\ephi v : V_G \mono V_H$ and $\ephi e : E_G \mono E_H$ are injective;
	\item[(LGE-2)\label{def:lge-2}]
		$\ephi i : X_G \mono \wp(X_H \uplus P_H)$ is fully
                injective: $\forall x\neq x' : \ephi i(x) \cap \ephi i (x') = \emptyset$;
	\item[(LGE-3)\label{def:lge-3}]
		$\ephi o : Y_G \to E_H \uplus Y_H$ in an arbitrary partial map;
	\item[(LGE-4)\label{def:lge-4}]
		$\rng(\ephi e) \cap \rng(\ephi o) = \emptyset$ and $\rng(\ephi{port})\cap \bigcup\rng(\ephi i)  = \emptyset$;
	\item[(LGE-5)\label{def:lge-5}]
		$\ephi p \circ \restr{\link_G^{-1}}{E_G} = \link_H^{-1}\circ \ephi e$;
	\item[(LGE-6)\label{def:lge-6}]
		$\ctrl_G = \ctrl_H \circ \ephi v$;
	\item[(LGE-7)\label{def:lge-7}]
		$\forall p \in X_G \uplus P_G : \forall p' \in (\ephi p)(p) : (\ephi h \circ \link_G)(p) = \link_h(p')$
	\end{description}
	where 
	$\ephi p \defeq \ephi i \uplus \ephi{port}$,
	$\ephi h \defeq \ephi e \uplus \ephi{o}$ and  
	$\ephi{port}:P_G \mono P_H$ is $\ephi{port}(v,i) \defeq (\ephi v(v),i))$.
\end{definition}
The first three conditions are on the single sub-maps of the 
embedding. Condition \ref{def:lge-4} ensures that no components 
(except for outer names) are identified; condition \ref{def:lge-5}
imposes that points connected by the image of an edge are all 
covered. Finally, \ref{def:lge-6} and \ref{def:lge-7} 
ensure that the guest structure is preserved i.e.~node controls 
and point linkings are preserved.

\paragraph{Place graph} Like link graph embeddings, place graph embeddings
are just a structure preserving injective map from nodes along with suitable 
maps for the inner and outer interfaces. In particular, a site is mapped to 
the set of sites and nodes that are ``put under it'' and a root is mapped to
the host root or node that is ``put over it'' splitting the host place graphs 
in three parts: the guest image, the context and the parameter (which are 
above and below the guest image).

\begin{definition}[Place graph embedding {\cite[Def~7.5.4]{hoesgaard:thesis}}]\label{def:pge}
	Let $G : n_G \to m_G$ and $H : n_H \to m_H$ be two concrete place graphs. 
	A \emph{place graph embedding} $\phi : G \emb H$ is a map
	$\phi \defeq \ephi v \uplus \ephi s \uplus \ephi r$
	(assigning nodes, sites and regions respectively)
	subject to the following conditions:
	\begin{description}\itemsep=0pt
	\item[(PGE-1)\label{def:pge-1}]
		$\ephi v : V_G \mono V_H$ is injective;
	\item[(PGE-2)\label{def:pge-2}]
		$\ephi s : n_G \mono \wp(n_H \uplus V_H)$ is fully injective;
	\item[(PGE-3)\label{def:pge-3}]
		$\ephi r : m_G \to V_H \uplus m_H$ in an arbitrary map;
	\item[(PGE-4)\label{def:pge-4}]
		$\rng(\ephi v) \cap \rng(\ephi r) = \emptyset$ and $\rng(\ephi v) \cap \bigcup \rng(\ephi s) = \emptyset$;
	\item[(PGE-5)\label{def:pge-5}]
			$\forall r \in m_G : \forall s \in n_G : \prnt_H^*\circ \ephi r(r) \cap \ephi s(s) = \emptyset$;
	\item[(PGE-6)\label{def:pge-6}]
		$\ephi c \circ \restr{\prnt_G^{-1}}{V_G} = \prnt_H^{-1}\circ \ephi v$;
	\item[(PGE-7)\label{def:pge-7}]
				$\ctrl_G = \ctrl_H \circ \ephi v$;
	\item[(PGE-8)\label{def:pge-8}]
				$\forall c \in n_G \uplus V_G : \forall c' \in \ephi c(c) : 
				(\ephi f \circ \prnt_G)(c) = \prnt_H(c')$;
	\end{description}
	where $\prnt_H^*(c) = \bigcup_{i < \omega} \prnt^i(c)$,
	$\ephi f \defeq \ephi v \uplus \ephi{r}$, and
	$\ephi c \defeq \ephi v \uplus \ephi{s}$.
\end{definition}

Conditions in the above definition follows the structure
of Definition~\ref{def:lge}, the main notable difference is
\ref{def:pge-5} which states that the image of a root
cannot be the descendant of the image of another. 
Conditions \ref{def:pge-1}, \ref{def:pge-2} and \ref{def:pge-3} 
are on the three sub-maps composing the embedding; conditions 
\ref{def:pge-4} and  \ref{def:pge-5} ensure that no components 
are identified; \ref{def:pge-6} imposes surjectivity on children 
and the last two conditions require the guest structure to be 
preserved by the embedding map.

\paragraph{Bigraph}
Finally, bigraph embeddings can now be defined as maps being composed 
by an embedding for the link graph with one for the place graph 
consistently with the interplay of these two substructures. In 
particular, the interplay is captured by a single additional 
condition ensuring that points in the image of an inner names 
reside in the parameter defined by the place graph embedding 
(i.e.~are inner names or ports of some node under a site image).

\begin{definition}[Bigraph embedding {\cite[Def~7.5.14]{hoesgaard:thesis}}]\label{def:bge}
	Let $G : \face{n_G,X_G} \to \face{m_G,Y_G}$ and 
	$H : \face{n_H,X_H} \to \face{m_H,Y_H}$ be two concrete
	bigraphs. A \emph{bigraph embedding} $\phi : G \emb H$
	is a map given by a place graph embedding 
	$\ephi P : G^P\emb H^P$ and a link graph embedding
	$\ephi L : G^L\emb H^L$ subject to the consistency
	condition:
	\begin{description}\itemsep=0pt
		\item[(BGE-1)\label{def:bge-1}]
			$\rng(\ephi i) \subseteq X_H \uplus 
			\{(v,i) \in P_H \mid \exists s \in n_G : k \in \mathbb{N}:
			\prnt_H^k(v) \in \ephi s(s)\}$.
	\end{description}
\end{definition}

\section{Partial and candidate partial bigraph embeddings}
\label{sec:pemb-cpemb}

In this Section we introduce the notion of \emph{partial bigraph
  embeddings}.  We show that for a given pair of guest and host
bigraphs, the set of their partial embeddings is endowed with a
``almost atomic'' meet-semilattice.  This structure will play a
central r\^ole in the algorithm presented in Section~\ref{sec:demb}.
We then consider also the situation when we know only a part of the
codomain of a partial embedding, by introducing the notion of
\emph{candidate partial embedding}.

\subsection{Partial bigraph embeddings}
\label{sec:pemb}
Basically,
a partial bigraph embedding is a partial map subject to the same
conditions of a total embedding (Definition \ref{def:bge}) up-to
partiality.

\begin{definition}[Partial bigraph embedding]\label{def:pbge}
	Let $G : \face{n_G,X_G} \to \face{m_G,Y_G}$ and 
	$H : \face{n_H,X_H} \to \face{m_H,Y_H}$ be two concrete
	bigraphs. A \emph{partial bigraph embedding} $\phi : G \pemb H$
	is a partial map subject, where defined, to the same conditions of Definition~\ref{def:bge}.
\end{definition}

As we will see in Section~\ref{sec:demb}, partial embeddings represent
the partial or intermediate steps towards a total embedding.  This
notion of ``approximation'' is reflected by the obvious ordering given
by the point-wise lifting of the anti-chain order to partial maps.
In particular, given two partial embeddings $\phi,\psi : G \pemb H$ we say that:
\begin{equation}
	\label{eq:def-emb-ord}
	\phi \sqsubseteq \psi \defiff \forall x \in \dom(\phi)\, 
	\phi(x) \neq \perp \implies \psi(x) = \phi(x)
\text{.}\end{equation}
This definition extends, for any given pair of concrete bigraphs $G$ and $H$,
to a partial order over the set of partial bigraph embeddings of
$G$ into $H$. It is easy to check that the entirely undefined embedding 
$\varnothing$ is the bottom of this structure and that meets are 
always defined:
\[
	\phi \sqcap \psi \defeq \lambda x.
	\begin{cases}
		\phi(x) & \text{if } \phi(x) = \psi(x) \\
		\perp & \text{otherwise} 
	\end{cases}
\]
Likewise, joins, where they exist, are defined as follows:
\[
	\phi \sqcup \psi \defeq \lambda x.
	\begin{cases}
		\phi(x) & \text{if } \phi(x) \neq \perp \\
		\psi(x) & \text{if } \psi(x) \neq \perp \\
		\perp &  \text{otherwise} 
	\end{cases}
\]
Clearly $\phi$ and $\psi$ have to coincide where are both defined and
their join $\phi \sqcup \psi$ is defined iff 
it does not violate any condition in Definition~\ref{def:pbge}.

The set of partial embeddings for a given guest $G$ and host $H$ is an
\emph{meet-semilattice}.  Moreover, an embedding can be represented as
the meet of a finite set of ``basic'' elementary partial embeddings,
i.e.~suitable elements from $G\times H$.  This suggests to use these
elementary partial embeddings as a compact representation for
(partial) embeddings.  Although elementary partial embeddings may
remind \emph{atomic elements} in meet-semilattices, they are not
really atomic. In fact, a partial embedding whose domain contains a
site (or an inner name) has to map it to the empty-set in order to be
minimal (and hence an atom); for this reason, a partial embedding
mapping a site to something different than $\emptyset$ could not be
described as the join of atoms.

This observation leads us to introduce the following definition.
\begin{definition}[(Almost) atomic partial embedding]
	\label{def:almost-atomic}
	A partial embedding $\alpha : G \pemb H$ is said to be
	\emph{(almost) atomic} whenever the following implication 
	holds true:
	\[
		\psi \sqsubseteq \alpha 
		\implies 
		\psi = \varnothing 
		\lor
		\exists!s \in n_G \uplus X_G.\psi(s) = \emptyset
		\text{.}
	\]
	The set of atoms below a partial embedding $\phi$
	is called \emph{base of $\phi$} and is denoted as
	$At(\phi)$.	The set of all atomic partial embeddings of
	$G$ into $H$ is denoted as $At_{G,H}$ (we shall drop the
	subscripts when confusion seems unlikely).
\end{definition}

\begin{proposition}[Base]
	\label{prop:almost-atomic}
	Let $\phi : G \pemb H$ be a partial embedding.
	There exists a minimal and finite family $At(\phi)$ of 
	(almost) atomic partial embeddings
	whose join is $\phi$.	
\end{proposition}
\begin{proof}
	Let $At(\phi)$ be the set of (almost) atomic partial embeddings 
	given by the union of:
	\begin{itemize}	
		\item
		$\left\{
			\restr{\phi}{u} \,\middle|\, 
			u \in V_G \uplus m_G \uplus E_G \uplus Y_G
		\right\}$,
		\item
		$\left\{
			\corestr{\phi}{w} \,\middle|\, 
			u \in n_G \uplus X_G \land w \in \phi(u)
		\right\}$, and
		\item
		$\left\{
			\restr{\phi}{u} \,\middle|\, 
			u \in n_G \uplus X_G \land \phi(u) = \emptyset
		\right\}$
	\end{itemize}
	where $\corestr{\phi}{w}$ denotes co-restriction.
	Then $\bigsqcup At(\phi) = \phi$ and 
	$\bigsqcup S \sqsubset \phi$ for any
	$S \subset At(\phi)$.
\end{proof}

\subsection{Candidate partial embeddings}
\label{sec:cpemb}
 A \emph{candidate partial embedding} is a partial map $\rho : G \pmap
H$ with the same domain and codomain of an embedding of $G$ into
$H$.  A \emph{candidate embedding} is a total map with
suitable domain and codomain. Note that every candidate defined only
on a single element is a partial embedding.

The notion of candidate partial embedding is accessory to the
decentralized algorithm we presents in Section~\ref{sec:demb}.  In
fact, families of partial embeddings are sent over the network as
graphs whose vertexes are atoms and whose edges represents admissible
joins. Joins are not transitive and some of the conditions of bigraph
embeddings cannot be checked by only looking at pairs of atoms and
their immediate neighbourhood, as we show in
Theorem~\ref{thm:small-checks} and Theorem~\ref{thm:big-checks}.

Before we present this result let us present \ref{def:lge-5} and
\ref{def:pge-6} in a more convenient (but equivalent) form,
that points out the conditions failing to be ``locally verifiable''.
\begin{description}\itemsep=0pt
	\item[(LGE-5a)\label{def:lge-5a}] 
		$\forall e \in E_G \forall x \in P_G \uplus X_G ( x \in \link^{-1}_G(e) \iff \ephi p(x)\subseteq \link^{-1}_H(\ephi e(e)))$
	\item[(LGE-5b)\label{def:lge-5b}] 
		$\forall e \in E_G \forall y \in \link^{-1}_H(\ephi e(e)) \exists x \in P_G \uplus X_G ( y \in \ephi p(x) )$
	\item[(PGE-6a)\label{def:pge-6a}] 
			$\forall v \in V_G \forall s \in n_G \uplus V_G ( s \in \prnt^{-1}_G(v) \iff \ephi c(s)\subseteq \prnt^{-1}_H(\ephi e(v)))$
	\item[(PGE-6b)\label{def:pge-6b}] 
			$\forall v \in V_G \forall y \in \prnt^{-1}_H(\phi^v(v)) \exists s \in n_G \uplus V_G ( y \in \phi^c(s) )$
\end{description}

\begin{theorem}\label{thm:small-checks}
	Let $\rho : G \to H$ be a candidate embedding and let
	$\alpha_1, \dots, \alpha_n$ the atoms forming it.
	$\rho : G \to H$ satisfies conditions
	\hyperref[def:lge]{(LGE-1-5a,6,7)} and \hyperref[def:pge]{(PGE-1-4,6a,7,8)}
	if, and only if,
        \vspace{-1ex}
	\begin{enumerate}[label=\em(\alph*)]\itemsep=0pt
		\item
			$\forall i$ $\alpha_i$ satisfies
			\ref{def:lge-3}, \ref{def:lge-6},
			\ref{def:pge-3}, and \ref{def:pge-7};
		\item
			$\forall i,j$ s.t.~the candidate $\alpha_i \sqcup \alpha_j$ 
			satisfies \hyperref[def:lge]{(LGE-1,2,4,5a,7)} and \hyperref[def:pge]{(PGE-1,2,4,6a,8)};
        \vspace{-1ex}
	\end{enumerate}
	and each check involves at most the components of $H$ adjacent
	to the image of $\alpha_i$ and $\alpha_j$.
\end{theorem}
\begin{proof}[Proof (Sketch)]
	Its easy the above conditions can be falsified by providing at most
	two atoms and that the negated formula of each condition involves
	at most one step along $\prnt_H$ or $\link_H$. As an example we detail 
	the case of \ref{def:lge-5a} leaving the others to the reader.
	If $\rho$ does not satisfy \ref{def:lge-5a}, then there are $e \in E_G$ 
	and $x \in P_G \uplus X_G$ s.t.:
	\begin{equation}\tag{$\dagger$}\label{def:nlge-5a}
		(x \in \link^{-1}_G(e) \land 
		\esf{\rho}{p}(x)\not\subseteq 
		\link^{-1}_H(\esf{\rho}{e}(e))) \lor (x \not\in \link^{-1}_G(e) \land \esf{\rho}{p}(x)\subseteq \link^{-1}_H(\esf{\rho}{e}(e)))
	\end{equation}
	Let $\bar{e} \in E_G$ and $\bar{x} \in P_G \uplus X_G$ two witnesses 
	of \eqref{def:nlge-5a} and consider the atomic partial embeddings 
	$\alpha_1 = \corestr{\rho}{\bar e}$ and $\alpha_2 = \corestr{\rho}{\bar x}$.
	Clearly $\alpha_1, \alpha_2 \sqsubseteq \rho$ and either
	$\bar{x} \in \link^{-1}_G(\bar{e}) \land \bar{y} \not\in 
	\link^{-1}_H(\bar{d})$ or $\bar{x} \not\in \link^{-1}_G(\bar{e})
	\land \bar{y} \in \link^{-1}_H(\bar{d})$.
\end{proof}

\begin{theorem}\label{thm:big-checks}
	Verifying whether a candidate satisfies conditions 
	\ref{def:bge-1}, \ref{def:lge-5b}, \ref{def:pge-5} and \ref{def:pge-6b}
	may require more than two atoms	or the neighbourhood of their images.
\end{theorem}
\begin{proof}[Proof (Sketch)]
	Conditions \ref{def:pge-5} and \ref{def:bge-1} contain the transitive closure of $\prnt_H$.
	Conditions \ref{def:lge-5b} and \ref{def:pge-6b}
	contain existential and universal quantifications at the same time.
\end{proof}

\begin{definition}
	\label{def:locally-checkable}
	\label{def:ancestor-checkable}
	Conditions \hyperref[def:lge]{(LGE-1-5a,6,7)} and
	\hyperref[def:pge]{(PGE-1-4,6a,7,8)} are called 
	\emph{locally checkable}, and the candidates satisfying 
	them are said \emph{locally checked}.
	Conditions \ref{def:pge-5} and \ref{def:bge-1} are called 
	\emph{ancestor checkable}, and the candidates satisfying them are said
	\emph{ancestor checked}.
\end{definition}

\section{State, overlay and reactions}
\label{sec:dbrs}

This section illustrates how a bigraph is distributed between a
processes family and how it is maintained and updated.  First, we
formalize the idea of a ``distributed bigraph'' and show how a
partition of the global system state defines a \emph{semantic overlay
  network}. The r\^ole of this network is crucial for the embedding
algorithm since communication will follow this structure.  Finally, we
describe how reactions are carried out concurrently and consistently.

In the following, let $\msf{Proc}$ denote the family of processes
forming the distributed bigraphical machine under 
definition and let $H$ be a generic concrete bigraph 
$(V_H, E_H, \ctrl_H, \prnt_H, \link_H):\face{n_H,X_H}\to\face{m_H,Y_H}$
over a given signature $\Sigma$.

\subsection{State partition}
\label{sec:dstate}
Intuitively, a partition of the shared state $H$ is a map
assigning each component of the bigraph $H$ to the process in charge of maintaining it.

\begin{definition}[State partition]\label{def:state-part}
	A partition of (the shared state) $H$ over $\msf{Proc}$
	is a map $\prt P : H \to \msf{Proc}$ assigning each component of $H$ to some process.
	In particular, $\prt P$ is given by the (sub)maps $\eprt{v}$, $\eprt{e}$, $\eprt{s}$, $\eprt{r}$, $\eprt{i}$, and $\eprt{o}$ on vertices, edges, sites, roots, inner names, and outer names respectively.
	Every component of $H$ in the pre-image of a process is said to be \emph{held by} or \emph{local to} that process.
	Ports are mapped into the process holding their node i.e.~$\prt P((v,i)) \defeq \prt P(v)$.
\end{definition}
State partitions define a notion of \emph{locality} or \emph{ownership} for bigraphs distributed across the given family of processes by a partition.
This notion extends directly to embeddings.
\begin{definition}[Local partial embedding]\label{def:local-emb}
	Let $\phi : G \pemb H$ be a partial embedding and let 
	$\prt P : H \to \msf{Proc}$ be a partition. 
	The \emph{owners} of $\phi$ are the processes in $\rng(\prt P \circ \phi)$.
	If $\phi$ has exactly one owner then it is said to be \emph{local} to it.
	We denote the restriction of $\phi$ to the
	portion of bigraph held by a set of processes $S \subseteq \mathsf{Proc}$ 
	as $\embrestr[\prt P]{S}{\phi}$; 
	we shall drop the partition $\prt P$ when confusion seems unlikely.
\end{definition}
Given a process $Q$, every partial embedding 
$\psi \sqsubseteq \embrestr[\prt P]{\{Q\}}{\phi}$
is local to $Q$--except for the undefined embedding $\varnothing$ 
since the set $\rng(\prt P \circ \varnothing)$ will always be empty.
Therefore, the set of atoms below the restriction of $\phi$ to $Q$
\[
	At\left(\embrestr{\{Q\}}{\phi}\right) = \left\{
		\alpha \in At 
	\,\middle|\, 
		\alpha \sqsubseteq \embrestr{\{Q\}}{\phi}
	\right\}
\] 
can be thought as the \emph{support of $\phi$ local to $Q$}; any
change in the bigraph held by $Q$ that affects one of these atoms
will necessarily invalidate $\phi$. This last observation is at 
the hearth of the retraction phase of the embedding algorithm 
(cf.~Section~\ref{sec:demb}).

The notion of adjacency for bigraph components lifts
to the family of processes along the given partition map. Here
hyper-edges of the link graph are considered as trees
where the root is the hyper-edge handle (i.e.~an edge or an outer name)
and leaves are all the points (i.e.~ports or inner names) it connects.
\begin{definition}\label{def:proc-adjto}
Let $Q,R \in \msf{Proc}$. The process $Q$ is said to be \emph{adjacent (w.r.t.~the partition $\prt P$) to} $R$ whenever one of the following holds:
\begin{description}\itemsep=0pt
	\item[(ADJ-P)\label{def:adj-p}] there exists a node, port or site $c$ s.t.~$\prt P(c) = Q$ and $\prt P(\prnt_H(c)) = R$;
	\item[(ADJ-L)\label{def:adj-l}] there exists a point $p$ s.t.~$\prt P(p) = Q$ and $\prt P(\link_H(p)) = R$;
	\item[(ADJ-T)\label{def:adj-t}] there exist two roots or handles $t,t'$ s.t.~$\prt P(t) = Q$ and $\prt P(t') = R$;
\end{description}
A partial embedding $\phi$ is said to be adjacent to a process $R$
(w.r.t.~$\prt P$) iff its image is. \vspace{-1.3ex}
Adjacency of $Q$ or $\phi$ to $R$ w.r.t. $\prt P$ is denoted by
${Q \stackrel{\raisebox{-2ex}{$\scriptscriptstyle\, \prt P$}}{\adjto} R}$
and ${\phi \stackrel{\raisebox{-2ex}{$\scriptscriptstyle\, \prt P$}}{\adjto} R}$
respectively (with the option to from $\prt P$ when confusion no confusion may arise). 
\end{definition}

The adjacency relation defines a directed graph with vertices in $\msf{Proc}$
and hence a directed \emph{overlay network} $\msf N_\prt P$. 
This network bares a specific semantic meaning because it reflects adjacency 
of the bigraphical elements held by each process forming the network: 
two processes are adjacent if, and only if, they hold components that
are adjacent in the distributed bigraph $H$. The network $\msf N_\prt P$
is such that shortest paths connecting processes in it cannot exceed
in length shortest paths between the components of $H$ they hold.
\begin{lemma}
	Let $c_1, c_2 \in H$. The length of shortest path in $\msf N_\prt P$
	connecting $\prt P(c_1)$ and $\prt P(c_2)$ is limited from above ed by the length
	of the shortest path in $H$ connecting $c_1$ and $c_2$.
\end{lemma}
\begin{proof}[Proof (sketch)]
	Definition~\ref{def:proc-adjto} characterizes the quotient
	induced by $\prt P$ on $H$.
\end{proof}
The last observation is crucial to our 
purposes since relates routing through the overlay $\msf N_\prt P$ with walks 
and visits of $H$ used e.g.~to compute embeddings into $H$ in non-distributed
settings. Notice that the restriction of $\msf N_\prt P$ to $\rng(\prt P)$
will always be connected i.e. for any two processes in $\rng(\prt P)$ 
there (at least) two paths starting from them and ending in the same node.
This ensures that there is always a ``rendezvous'' point for two messages
(and in particular two partial embeddings to be combined).
Connectedness is ensured by \ref{def:adj-t} but this condition is sufficient
and can be relaxed by assuming the adjacency relation to contain a 
directed-complete partial order (dCPO) on $\rng(\prt P)$.
Note that each process is aware to its neighbouring processes and the nature
of their adjacency because each process knows parents, children, etc. of each
component it hold.

\begin{remark}
	In \cite{mpm:gcm14w} we considered, for the sake of simplicity, an 
	undirected graph as overlay network. However, the additional information
	of a directed overlay network allows for more efficient routing strategies
	hence reducing duplicated computations of partial embeddings 
	(cf.~Section \ref{sec:demb}). 
	In fact, edge direction reflects the structure of the bigraph
	and can be leveraged also by partition strategies to
	distribute the bigraph privileging locality of reactions.
\end{remark}

\begin{example}[Multi-Agent Systems]
	\looseness=-1
	In \cite{mmp:dais14} we described how BRS can be used to both design
	and prototype multi-agent systems (MAS). In loc.~cit. BRS are used 
	to model the application domain lending helpful formal 	verification 
	tools (e.g.~model checkers) to the designer as long as simulation 
	ones. Then entities forming each bigraph are divided as \emph{subjects} and 
	\emph{objects} accordingly to their r\^ole in the model (e.g.~node controls);
	with the former	being the agents in the systems. When agents are identified
	with processes of a D-BAM this yield a prototype of the system where agent
	cooperation and reconfiguration correspond to negotiation of execution
	strategies and reactions respectively.
	
	In \cite{mmp:dais14} each entity designated as object 
	(e.g.~a node modelling a good) is assigned to the process of its first 
	ancestor designated as a subject (e.g.~a node modelling a store). 
	This is an instance of \emph{partition strategy}.
	In particular, the partition is driven by the application domain privileging
	locality of interactions: a store is going to be involved by each reaction affecting its goods.
\end{example}

\subsection{Distributed reactions}
\label{sec:dreac}
Let $\phi$ be an embedding of $G$ into the bigraph $H$ distributed across
the processes in the system and let $r: G \rightarrowtriangle G'$ 
be a parametric rewriting rule for the given BRS. 
Processes holding elements of $G$ image through $\phi$ 
or in its parameters have to negotiate the firing of $r$ 
and coordinate the update of their state. The negotiation
phase is related to the specific execution policy and hence
is left out from the present work (see \cite{mmp:dais14,perrone:thesis} for
an example). The update phase involves a distributed
transaction is handled by established algorithms
like \emph{two-phase-commit} \cite{cooper1982:distcommit}.

Each process concurrently enacts two roles: one active and one passive.
In the first case:
\begin{enumerate*}[label=\em(1\alph{*})]
	\item 
		it selects a reaction 
		(e.g.-rewriting rule, edit script) and 
		a suitable embedding among those provided by
		its embedding engine;
	\item
		starts a transaction with all the processes
		involved in the embedding (i.e.~$\rng(\prt P \circ \phi)$);
	\item
		waits for them to either approve or reject the reaction and 
		completes the transaction protocol accordingly.
\end{enumerate*}
In the second case:
\begin{enumerate*}[label=\em(2\alph{*})]
	\item 
		it waits for other processes to propose a reaction;
	\item
		votes for acceptance or rejection (execution strategy);
	\item
		executes the reaction iff each other participant 
		agrees on committing the transaction.
\end{enumerate*}
Note that consistency of the current bigraph is guaranteed by the correctness of the
distributed transaction protocol, even in presence of outdated
embeddings or concurrent transactions.

In \cite{mmp:dais14} reactions correspond to agent reconfigurations.
These may result in agent creation or termination requiring a 
\emph{life-cycle} for processes of the D-BAM too--since the latter are
identified with the former. Although we assumed a fixed
family of processes, to simplify the exposition, the D-BAM 
supports churns that are contextual to reactions,
especially when partitions are implicitly adapted by partition
strategies of the like of \cite{mmp:dais14}.

\section{Distributed embedding}
\label{sec:demb}
In this Section we introduce a decentralized algorithm for computing bigraphical embeddings
in the distributed settings outlined in Section~\ref{sec:dbrs}
and Figure~\ref{fig:dbrs-arch}.
Intuitively, each process running this algorithm maintains a 
private collection of partial embeddings for the guests it has to look
for and cooperates with its neighbouring processes to complete 
or refute them. 

For the sake of simplicity we assume that all processes are given 
the same set of guests (e.g.~the redexes of parametric rewriting rules defining the BRS
being executed by the D-BAM), that this set is fixed over the 
time and does not contain the empty bigraph. However, these mild
assumptions can be dropped with minor changes to the algorithm.
Likewise, we assume causally ordered
communication and refer the reader to \cite{mpm:gcm14w} for a version
of the algorithm where message causality and group communication are 
explicitly implemented on reliable point-to-point channels by means 
suitable logical clocks (i.e.~internal counters that every process 
attach to the information it generates).

\newcommand{\msg}[1]{\langle #1 \rangle}

\SetAlgoCaptionSeparator{.}
\NoCaptionOfAlgo
\DontPrintSemicolon
\newcommand\mycommfont[1]{\footnotesize\textcolor{blue}{#1}}
\SetCommentSty{mycommfont}
\SetKwFunction{img}{img}
\SetKwFunction{send}{send}
\SetKwFunction{sto}{to}
\SetKwFunction{self}{self}
\SetKwFunction{sug}{suggest}
\SetKwFunction{ret}{retract}
\SetKwFunction{update}{onUpdate}
\SetKwFunction{LAECalc}{getLocalAtoms}
\SetKwFunction{publish}{onSuggest}
\SetKwFunction{updateOverlay}{waitOverlayUpdate}
\SetKwFunction{PrEmb}{getCandidateEmbeddings}
\SetKwFunction{addEmbedding}{addEmbedding}
\SetKwFunction{removeEmbeddings}{removeEmbeddings}
\SetKwFunction{withhold}{onRetract}

\subsection{Computing and updating partial embedding}
Each process $Q$ in the D-BAM executes the \emph{embedding engine} module
alongside the reaction engine (cf.~Figure~\ref{fig:dbrs-arch}) with which
it asynchronously communicates by means of shared state structures.
On one side, the module observes the chunk of the current bigraph
held by the process and the updates the reaction module commits on it;
this defines the input of the reaction engine. (Note that overlay network 
$\mathsf N_\prt P$ are implicitly and consistently
updated during each distributed transaction wrapping a reaction.)
On the other side, the module provides a collection of \emph{available embeddings}
i.e.~a partial view of all the embeddings computed by the machine.
This defines the output of the module. 
Although processes often have an incomplete view, the algorithm guarantees
that each embedding is computed by at least one of them.

Reactions may invalidate embeddings which then have to be collected
by this module. 
Each embedding engine operates on its local collection of available embeddings 
by means of two procedures:	
$\addEmbedding(\phi)$ and $\removeEmbeddings(\Psi)$ where the
second removes all embeddings $\phi$ s.t.~$\psi \sqsubseteq \phi$ 
for some $\psi \in \Psi$. 
High consistency of available embeddings collections is not 
mandatory (reactions are consistent) allowing us
to trade some of it for performance and adopt an asynchronous
garbage collection scheme for sweeping invalidated embeddings.

An embedding may be owned by more than one process forcing their
execution engines to exchange information in order to 
compute/invalidate it. 
The data being exchanged consists of suggestions or retractions of
partial embeddings and is conveyed by two kind of messages: \sug and \ret. 
The former kind push newly discovered partial embeddings to other 
processes and the latter propagate invalidations.
For efficiency reasons, partial embeddings are sent in batches 
encoded as irreflexive undirected graphs (called \emph{atom graphs}) 
whose nodes are the atoms composing them 
(cf.~Proposition~\ref{prop:almost-atomic}) and whose edges
are checkable joins in the sense of Theorem~\ref{thm:small-checks}.
Atom graphs implicitly describe candidates but, by Theorem~\ref{thm:big-checks}
embeddings cannot be singled out without looking at more than
two atoms or their images; information that is available at suitable
stages of the algorithm only.

The same encoding is used by each process to store the set of
(candidate) partial embeddings forming its partial view of those
existing in the system. To simplify the exposition we assume this 
structure as indexed over the set of guests (hence duplicating information
relative to their overlaps).
We shall denote this structure by $\Gamma_{Q,G}$,
where $Q$ is the owning process and $G$ is the guest bigraph,
and drop the subscripts when clear from the context.
Each process $Q$ implicitly keeps track of which processes
it received an atom from; this set will be denoted as 
$\src_Q(\alpha)$.

\begin{algorithm2e}[t]
\SetAlgoRefName{onRetract}
\caption{}\label{algo:retract}
\SetKwProg{myproc}{Event handler}{}{}
\myproc{\withhold{$G$, $RA$, $RE$}}{
$(A,E) \gets \Gamma_{Q,G}$\;
$(A',E') \gets (A\setminus RA, E\setminus RE)$\;
\If{$A \not= A'\lor E \not= E'$}
{
$\Gamma_{Q,G} \gets (A', E')$\;
\removeEmbeddings{$RA \cup \{\alpha \sqcup \alpha' | \{\alpha,\alpha'\} \in RE\}$}\;
\send \ret$\msg{G, RA, RE}$ \sto $\{ P \mid {Q \adjto P}\}$
}
}
\end{algorithm2e}
\begin{algorithm2e}[t]
\SetAlgoRefName{onSuggest}
\caption{}\label{algo:suggest}
\SetKwProg{myproc}{Event handler}{}{}
\myproc{\publish{$G$, $A'$, $E'$}}{
$(A,E) \gets \Gamma_{Q,G}$\;
$A'' \gets A \cup A'$\;
$E'' \gets E \cup E' \cup$\{$\{\alpha,\alpha'\} \mid 
\alpha \sqcup \alpha' \in A\bigsqcup A' \text{ is \hyperref[def:locally-checkable]{locally checked} and \hyperref[def:ancestor-checkable]{ancestor checked}}$\}\;
\If{$A \not= A'' \lor E \not= E''$}
{
$\Gamma_{Q,G} \gets (A'' , E'')$\;
\For{$\phi \in \PrEmb{A'', E''}$}{
	\If{$\phi$ \textnormal{satisfies \ref{def:lge-5b} and \ref{def:pge-6b}}}{
		\addEmbedding{$\phi$}
	}
}
\send \sug$\msg{G, A',E'}$ \sto $\{ P \mid {Q \adjto P}\}$}
}
\end{algorithm2e}

Writes on $\Gamma_{Q,G}$ are triggered by $Q$ receiving
\ret or \sug messages. The two events are handled by
\hyperref[algo:retract]{\withhold} and 
\hyperref[algo:suggest]{\publish} respectively.
Retractions remove from $\Gamma_{Q,G}$
all invalidated atoms and edges--note that these
are collections, not an actual graph. If
any change is made the information if propagated to 
the neighbourhood of $Q$ and to the collection of
available embeddings resulting in the removal of 
embeddings incoherent with the current bigraph $H$.
Likewise suggestions add new atoms and locally
checked joins to $\Gamma_{Q,G}$ being these 
edges in the message payload $E'$ or computed
by $Q$ from its view of the bigraph (recall that
every process knows parents, children, etc. of
every component it holds). 
Whenever changes to $\Gamma_{Q,G}$ are made,
these are propagated to the process neighbourhood.
Contextually, candidate embeddings (i.e.~cliques in
$\Gamma_{Q,G}$ whose atoms cover $G$ with their domains)
are checked to single out any new embedding to be added to the
collection of available ones. 
All \hyperref[def:locally-checkable]{locally}
and \hyperref[def:ancestor-checkable]{ancestor} checkable
conditions are encoded as edges leaving
\ref{def:lge-5b} and \ref{def:pge-6b} 
to be checked right before executing \addEmbedding.
Ancestor checkable conditions require some extra care since the
transitive closure of the place graph is involved. In general,
processes have only a partial view of $\prnt_H^*$ but this
is sufficient under mild conditions on how atoms for roots,
sites and inner names of $G$ routed. In fact, if this kind of atoms are travel
along $\prnt_H$ then, the \emph{least ancestor} of their images 
(cf.~Lemma~\ref{lem:ancestor-checkable}) can check \ref{def:pge-5} 
and \ref{def:bge-1} by knowing the source of the message containing them 
(besides its atom graph and the one in the message).

\begin{algorithm2e}[t]
\SetAlgoRefName{onUpdate}
\caption{}\label{algo:update}
\SetKwProg{myproc}{Event handler}{}{}
\myproc{\update{}}{
\For{$G\in Guests$}{
	$(A,E) \gets \Gamma_{Q,G}$\;
	$(A',E') \gets$ \LAECalc{G}\;
	$RA \gets \{\alpha \in A\setminus A' \mid \rng(\prt P \circ \alpha) = \{Q\}\}$\;
	$RE \gets$ $\{\{\alpha_1,\alpha_2\} \in E \mid 	\alpha_1 \in A \cap A' \land (\alpha_2 \in A \cap A' \implies \{\alpha_1,\alpha_2\} \notin E')\}$\;
	\send \ret$\msg{G, RA, RE}$ \sto \self
	\;
	\updateOverlay{}\;
 	\send \sug$\msg{G, A'\setminus A, E' \setminus E}$ \sto \self
}
}
\end{algorithm2e}

The mechanism offered by \hyperref[algo:retract]{\withhold} and 
\hyperref[algo:suggest]{\publish} is also used by the event handler 
\hyperref[algo:update]{\update} to propagate the effect of reactions
involving $Q$ to $\Gamma_{Q,G}$ and the rest of the system.
The handler is triggered during the commit phase of any write to 
the partial view of the current bigraph owned by $Q$ and computes 
the ``effect'' of the write by looking for changes in the graph of 
atoms local to $Q$.
The new graph can be computed applying the algorithm described in 
\cite{mp:memo14} (with minor adaptations to restrict the solution 
to atomic partial embeddings only).
Then, the graph is compared to $\Gamma_{Q,G}$ 
(note that $\Gamma_{Q,G}$ may contain also atoms local to other processes)
to find atoms and edges that have to be added or removed. 
Changes are passed to \hyperref[algo:retract]{\withhold} and
\hyperref[algo:suggest]{\publish}. Note that
propagation of retracts to processes involved in the update
has to be completed before any change to the overlay network 
is applied (i.e.~between transaction commit approval and finalization)
since this allows retracts to be dispatched along the same
route of the atoms they are collecting. Concurrent reaction may still
prevent every invalidated atom to be collected by this mechanism,
however consistency of the machine state is still preserved by
reactions being wrapped by distributed transaction. 
Another
viable approach is offered by remote references and leasing times:
atoms whose leasing is not renewed are considered retracted and 
automatically removed from the system. However, more messages
would be exchanged in order to periodically renew leasing times.

\subsection{Enhancements and heuristics}
\label{sec:heu}

\paragraph{Routing}
To simplify the presentation of the algorithm suggestions and retractions
are sent indistinctly to the entire neighbourhood resulting in
part of them being discarded by receivers.
In particular, candidates that are not adjacent to a receiver
are always discarded since the receiving process cannot
contribute to or benefit from them in any way.

Therefore, atom graphs have to parted and dispatched only
to those process adjacent to the candidates they describe.
Formally, an atom graph is adjacent to a process whenever
it can be covered by cliques each containing an atom adjacent
to the process.
\begin{definition}
	\label{def:adj-routing}
	An atom graph $(A,E)$ is said to be \emph{adjacent} to a process $Q$
	if, and only if, there exists a family of cliques $\{(A_1,E_1),\dots,(A_k,E_k)\}$	such that:
	\begin{itemize}\itemsep=0pt
		\item 
			$(A,E) = \bigcup_{i} (A_i,E_i)$;
		\item
			there is $\alpha \in A_i$ s.t.~$\alpha \adjto Q$
			for each $1 \leq i \leq k$;
		\item
			for each $\alpha \in A$,
			if $\dom(\alpha) \in m_G \uplus n_G$ then
			$\alpha \adjto Q$.
	\end{itemize}
\end{definition}
Adjacency based routing is handled at the communication level,
like causal ordering of messages. which sends to each recipient
of a multicast send only the greatest sub-graph adjacent to it.
Henceforth, we assume messages to be parted and dispatched
following this routing protocol.

\paragraph{Isomorphisms}
The network footprint of the algorithm suffers from combinatorics
due to internal isomorphisms of guest and host bigraphs
(cf. Theorem \ref{thm:traffic}). Here we suggest an heuristic
aimed to mitigate the impact of this phenomenon.

Consider the relation on atomic partial embeddings defined,
for any two $\alpha,\beta \in At_{G,H}$, as:
\[
	\alpha \equiv \beta 
	\defiff 
	\alpha \cong \beta 
	\text{ and }
	\forall \gamma \in At_{G,H}\setminus\{\alpha,\beta\}
	\alpha \sqcup \gamma \iff \beta \sqcup \gamma
\]
where $\alpha\cong \beta$ whenever there are two
bigraph isomorphisms $\sigma_G$ and $\sigma_H$ s.t.~$
	\alpha \circ \sigma_G = \sigma_H \circ \beta
	$.
It is easy to check that this relation is an equivalence
and hence defines quotients for atom graphs i.e.~an effective
compression for messages and, in general, structures based
on atom graphs.
A lossless compression requires atoms bo be decorated with their
multiplicity (and any list of additional user provided
properties often found in some extensions of bigraphs).

\subsection{Adequacy}
Reactions change the current bigraph and can be though as resetting
the embedding engine with the latter then checking and updating its
state coherently.
Reworded, reactions are perturbations the embedding engine has to 
stabilize from and restoring the equilibrium produces traffic over 
the network. Traffic stops only when the equilibrium is reached 
i.e.~the machine stabilizes.
\begin{theorem}[Completeness]
	When the system is stable, every embedding can
	be found in the collection of available embeddings
	of some process.
\end{theorem}

By causally ordered communication we can assume, w.l.o.g.,
that the system stabilized before the last reaction. Then
completeness is equivalent to the fact that for each $\phi : G \emb H$
there is some $Q$ s.t.~$\phi \in (\Gamma_{Q,G})^\star$
where $(\Gamma_{Q,G})^\star$ is the set of partial embeddings
whose atoms are in $\Gamma_{Q,G}$
\begin{lemma}
	If $\phi$ is a partial embedding for $G$ then 
	there is a process $Q$ s.t.~$\phi \in (\Gamma_{Q,G})^\star$.
\end{lemma}
\begin{proof}
	The proof is given by induction on the size of
	$\{Q_1,\dots,Q_k\} = \rng(\prt P\circ \phi)$.
	If $k = 1$ then the embedding is local to $Q_1$ and hence
	$\phi \in (\Gamma_{Q,G})^\star$.
	Otherwise, let $\phi_i = \embrestr{\phi}{\{Q_i\}}$ for
	$1 \leq i \leq k$. By inductive hypothesis each
	$\phi_i \in (\Gamma_{Q_i,G})^\star$.
	By connectedness hypothesis there is at least one process 
	$Q$ reachable by each $Q_i$.
	Messages are routed to all, and only, 
	the processes that can benefit from or contribute to them,
	in particular to $Q$. All edges that are locally checked
	and ancestor checked are added while messages travel the network.
	We only have to prove that there is always a process that can 
	add each edge along the paths to $Q$. By Theorem~\ref{thm:small-checks}, the only cases left are ancestor checkable. We conclude by
	Definition~\ref{def:adj-routing} and by
	Corollary~\ref{cor:ancestor-checkable}.
\end{proof}

\begin{lemma}
	\label{lem:ancestor-checkable}
	Let $r \in m_G$, $s \in n_G$, 
	$\alpha : r \mapsto u$, $\alpha' : s \mapsto \{u'\}$ two
	atoms, and $v,v'$ be the roots above $u$ and $u'$ respectively. 
	If $Q$ is the process to receive/compute $\alpha$ and $\alpha'$ earlier then
	at least one of the following is true:
	\begin{enumerate}[label=(\alph*)]\itemsep=0pt
		\item
			$Q$	holds the least ancestor of $u$ and $u'$;
		\item
			$Q$ holds both $v$ and $v'$;
		\item
			$Q$ holds either $v$ or $v'$ and the
			process holding the other sent the embedding.
	\end{enumerate}
	Let $i \in X_G$, $(u'',p) \in V_H$, $1 \leq p \leq \ar\circ\ctrl_H(u'')$,
	and $\alpha'' : i \mapsto \{(u'',p)\}$. There is a process $Q$ 
	that holds $u'$ and an ancestor of $u''$.
\end{lemma}
\begin{proof}[Proof (Sketch)]
	Atoms for guest sites and roots are dispatched following $\prnt_H$ only.
	Atoms for host ports are dispatched following both $\prnt_H$ and $\link_H$.
\end{proof}

\begin{corollary}[Ancestor checks]
	\label{cor:ancestor-checkable}
	For any two ancestor checkable atoms involving host ports, guest roots or sites there is a process that computes their edge before the system
	stabilize.
\end{corollary}
\begin{proof}[Proof (Sketch)]
	The process receiving/computing the atoms for guest sites and roots
	earlier checks them by looking at his piece of the shared bigraph
	and at the adjacency witness used to dispatch the message (i.e.~which child or sibling root was used by the sender process to route the message).
	Likewise, a process holding the image of a site checks whether a received
	inner name sits below it.
\end{proof}

\begin{theorem}[Soundness]
	If the system stabilizes then each embedding
	in the collection of available ones is valid w.r.t.~the
	current bigraph.
\end{theorem}
\begin{proof}
	Effects of reactions are computed locally to each 
	embedding engine and then propagated through the network.
	Propagation stops as soon as it stops producing changes
	in each $\Gamma$. By network connectedness and stabilization
	of the machine each invalid	embedding is eventually computed
	and removed by \hyperref[algo:retract]{\withhold}. Embeddings are added only by
	\hyperref[algo:suggest]{\publish} which filters out candidate
	and partial embeddings.
\end{proof}

\subsection{Complexity} 

The arity of the set of all embeddings of $G$ into $H$ is 
in $\mathbf O(|G|^{|H|})$ since, in the worst case, guest and host
encode two finite sets with a root for each element. On the other hand,
by Proposition~\ref{prop:almost-atomic}, the same set is described by
families in $At_{G,H}$ or, following the representation used by the algorithm,
by a suitable graph on $At_{G,H}$. Because elements of $At_{G,H}$
are essentially pairs from $G \times H$ the spatial
complexity of the graph representation is in $\mathbf O(|G|^2\cdot|H|^2)$
without any particular encoding. The same bound holds for the size of
each message sent on the overlay network. 
However, a process sends over the network
only nodes and edges it adds or removes from his $\Gamma_G$ and messages
are dispatched on the base of their semantic adjacency. Therefore,
between two reactions, every edge travels a link at most once
(either inside a suggest or retract message). 

\begin{lemma}
	The number of links in  $\mathsf N_\prt P$ is in $\mathbf O(|H|)$.
\end{lemma}
\begin{proof}
	The number $L$ of links in $\mathsf N_\prt P$ is bounded
	by the size of $H$ since the $\mathsf N_\prt P$ is a quotient
	of $H$. Hence, the worst case network is $\mathsf N_\prt I$
	where $\prt I : H \mono \mathsf{Proc}$ is the finest possible
	partition (i.e~each component is assigned a distinct process).
	Except for the clique induced by roots and handles,
	$\mathsf N_\prt I$ is a directed acyclic graph where
	each vertex has at most 
	$1 + \max_{v \in V_H }(\ar \circ \ctrl_H)(v)$ 
	outgoing edges and therefore is bounded by the maximal arity $k$
	occurring the given signature $\Sigma$ which is a fixed parameter
	of the D-BAM, hence a constant.
	The remaining case is given by the clique of roots and handles;
	their outgoing degree may exceed $k$ but their topology 
	can be easily reorganized to into a tree that satisfies the 
	bound and the above reasoning.
	Therefore, $L$ is bounded by the number of components of $H$.
\end{proof}

The algorithm generates, in the \emph{worst case scenario}, as much
traffic as a centralized one in its \emph{best case scenario}.
\begin{theorem}
	\label{thm:traffic}
	The traffic generated over $\mathsf N_\prt P$ while finding all the 
	available embedding,
	between two reactions, is in $\mathbf O(|G|^2\cdot|H|^3)$.
\end{theorem}
This scenario corresponds to bigraphs and partitions forcing information to
traverse all the network. In fact, the algorithm sends atoms only to
processes that can effectively benefit from it and hence their propagation
is stopped as soon as possible while retaining correctness and completeness.

In a typical scenario guests are fixed over time (hence a constant) and
$|H|$ outmatches $|G|$ by orders of magnitude. Moreover, embeddings
unaffected by a reactions are not recomputed.

\section{Conclusions and future work}\label{sec:concl}

In this paper we have presented a D-BAM, an abstract machine for
executing BRSs in a distributed environment.  The core novelty of
this machine is an algorithm for computing bigraph embeddings in a
distributed environment where the host bigraph is spread across
several cooperating peers.  Differently from existing algorithms
\cite{gdbh:implmatch,mp:memo14,sevegnani2010sat}, this one is
completely decentralized and does not to have a complete view of the
global state in any process in the system; hence it can scale to
handle bigraphs too large to reside on a single process/machine.

As in any distributed system, the complexity of our algorithm is
rendered by the number and the size of exchanged messages (i.e., the
\emph{network footprint}).  On one hand, the number of messages needed
for computing an embedding is linearly bounded by the size of the
embedded bigraph (which usually is constant during execution) and the
depth of the parent map of the host.  The worst case
(Theorem~\ref{thm:traffic}) is when the overlay network of processes
is a list, and atoms have to traverse it entirely. This case happens
for bigraphs and embeddings that can be seen as ``pathological'' in
the context of BRS; this suggests to consider different encodings of
the model into the BRS in order to improve locality of reactions.  On
the other hand, the size of messages depends on internal isomorphisms
in the guest and host bigraphs: these symmetries yield a combinatorial
explosion of the possible embeddings, leading to larger messages to be
exchanged between processes. This is mitigated by the heuristics
presented in Section~\ref{sec:heu}.  A possible future work is to
perform a formal analysis of locality and isomorphisms and their
impact in the context of smoothed complexity.

When a reaction is applied, it alters the distributed state and
inherently invalidates some of the partial embeddings computed by each
process.  Consistency of the state is guaranteed by reactions being
wrapped inside distributed transactions, but invalidated embeddings
are an unnecessary burden. To this end, we used a retraction
mechanisms as an \emph{asynchronous distributed garbage collection};
moreover, embeddings that are not affected by a reaction are not
recomputed.  We think that this approach is a good trade-off between
performance and consistency.  In fact, other solutions can be
implemented; for instance, invalidated embeddings can be collected
during the reaction commit phase; this offers the highest consistency
(the set of available embeddings will never contain invalid ones) at
the cost of slower reactions. On the other extreme of the spectrum,
invalidated embeddings are collected only when an inconsistency is
found by some process. Reactions are as fast as in presence of
asynchronous retractions but process data structures are heavily
polluted by invalid embeddings resulting in a higher rate of
aborted transactions i.e.~failed reactions.

An interesting feature of the bigraphical framework is that, given a
bigraph and a redex, we can calculate the \emph{minimal contexts}
(called IPOs)
needed to complete the bigraph in order to match the given redex.
Leveraging this property, a different, ``semi-distributed''
implementation of the bigraphical abstract machine has been proposed
in \cite{mmp:dais14}.  According to this algorithm, a process willing
to perform a rewrite has to
\begin{enumerate*}[label=\em(\arabic*)]
	\item
		collect a (suitable) view of the host
		bigraph from its neighbour processes; 
	\item 
		compute locally all the
		embeddings (i.e.~all possible reactions for 
		the given rewriting rule);
	\item
		apply the execution policy and start 
		a distributed rewriting inside a transaction.
\end{enumerate*} 
The existence of minimal contexts provide a bound to
the view a process has to collect at step 1. However, this bound is
outmatched by more substantial drawbacks, e.g.: parametric rules have
to be expanded into ground ones beforehand, and each process may end
up visiting (and copying) the entire bigraph.  Hence, we think that
the algorithm proposed in this paper outperformes the one in \cite{mmp:dais14}. 

A direct application of the distributed embedding algorithm is to
simulate, or execute, \emph{multi-agent systems}.  In
\cite{mmp:dais14} the authors propose a methodology for designing and
prototyping multi-agent systems with BRSs. Intuitively, the
application domain is modelled by a BRS and entities in its states are
divided as ``subjects'' and ``objects'' depending on their ability to
actively perform actions. Subjects are precisely the agents of the
system and reactions are reconfigurations.  This observation yields a
coherent way to partition and distribute a bigraph among the agents,
which can be assimilated to the processes of the distributed
bigraphical machine (execution policies are defined by agents desires
and goals).  Therefore, these agents can find and perform bigraph
rewritings in a truly concurrent, distributed fashion, by using the
distributed embedding algorithm presented in this paper.

Finally, we observe that the performance of the algorithm (and hence
of the D-BAM) depends on how the bigraph is partitioned and
distributed. For instance, it is easy to devise a situation in which
even relatively small guests require the cooperation of several
processes, say nearly one for each component of the guest. An
interesting line of research would be to study the relation between
guests, partitions, and performance in order to develop efficient
distribution strategies. Moreover, structured partitions lend
themselves to ad-hoc heuristics and optimizations.  As an example, the
way bigraphs are distributed among agents in \cite{mmp:dais14} takes
into account of their interactions and reconfigurations.


{\small
\begin{thebibliography}{10}

\bibitem{bgm:biobig}
G.~Bacci, D.~Grohmann, and M.~Miculan.
\newblock Bigraphical models for protein and membrane interactions.
\newblock In G.~Ciobanu, editor, {\em Proc.~MeCBIC}, volume~11 of {\em
  Electronic Proceedings in Theoretical Computer Science}, pages 3--18, 2009.

\bibitem{bmr:tgc14}
G.~Bacci, M.~Miculan, and R.~Rizzi.
\newblock Finding a forest in a tree - the matching problem for wide reactive
  systems.
\newblock In M.~Maffei and E.~Tuosto, editors, {\em Proc.~TGC}, volume 8902 of
  {\em Lecture Notes in Computer Science}, pages 17--33. Springer, 2014.

\bibitem{bdehn:fossacs06}
L.~Birkedal, S.~Debois, E.~Elsborg, T.~Hildebrandt, and H.~Niss.
\newblock Bigraphical models of context-aware systems.
\newblock In L.~Aceto and A.~Ing{\'o}lfsd{\'o}ttir, editors, {\em
  Proc.~FoSSaCS}, volume 3921 of {\em Lecture Notes in Computer Science}, pages
  187--201. Springer, 2006.

\bibitem{bghhn:coord08}
M.~Bundgaard, A.~J. Glenstrup, T.~T. Hildebrandt, E.~H{\o}jsgaard, and H.~Niss.
\newblock Formalizing higher-order mobile embedded business processes with
  binding bigraphs.
\newblock In D.~Lea and G.~Zavattaro, editors, {\em Proc.~COORDINATION}, volume
  5052 of {\em Lecture Notes in Computer Science}, pages 83--99. Springer,
  2008.

\bibitem{cooper1982:distcommit}
E.~C. Cooper.
\newblock Analysis of distributed commit protocols.
\newblock In {\em Proc.~SIGMOD}, pages 175--183. ACM, 1982.

\bibitem{dhk:fcm}
T.~C. Damgaard, E.~H{\o}jsgaard, and J.~Krivine.
\newblock Formal cellular machinery.
\newblock {\em Electronic Notes in Theoretical Computer Science}, 284:55--74,
  2012.

\bibitem{fph:gcm12}
A.~J. Faithfull, G.~Perrone, and T.~T. Hildebrandt.
\newblock {BigRed}: A development environment for bigraphs.
\newblock {\em ECEASST}, 61, 2013.

\bibitem{gdbh:implmatch}
A.~Glenstrup, T.~Damgaard, L.~Birkedal, and E.~H{\o}jsgaard.
\newblock An implementation of bigraph matching.
\newblock {\em IT University of Copenhagen}, 2007.

\bibitem{hoesgaard:thesis}
E.~H{\o}jsgaard.
\newblock {\em Bigraphical Languages and their Simulation}.
\newblock PhD thesis, IT University of Copenhagen, 2012.

\bibitem{jm:popl03}
O.~H. Jensen and R.~Milner.
\newblock Bigraphs and transitions.
\newblock In A.~Aiken and G.~Morrisett, editors, {\em POPL}, pages 38--49. ACM,
  2003.

\bibitem{kmt:mfps08}
J.~Krivine, R.~Milner, and A.~Troina.
\newblock Stochastic bigraphs.
\newblock In {\em Proc.~MFPS}, volume 218 of {\em Electronic Notes in
  Theoretical Computer Science}, pages 73--96, 2008.

\bibitem{mmp:dais14}
A.~Mansutti, M.~Miculan, and M.~Peressotti.
\newblock Multi-agent systems design and prototyping with bigraphical reactive
  systems.
\newblock In K.~Magoutis and P.~Pietzuch, editors, {\em Proc.~DAIS}, volume
  8460 of {\em Lecture Notes in Computer Science}, pages 201--208. Springer,
  2014.

\bibitem{mpm:gcm14w}
A.~Mansutti, M.~Miculan, and M.~Peressotti.
\newblock Towards distributed bigraphical reactive systems.
\newblock In R.~Echahed, A.~Habel, and M.~Mosbah, editors, {\em Proc.~GCM'14},
  page~45, 2014.
\newblock Workshop version.

\bibitem{mp:br-tr13}
M.~Miculan and M.~Peressotti.
\newblock Bigraphs reloaded: a presheaf presentation.
\newblock Technical Report UDMI/01/2013, Dept.~of Mathematics and Computer
  Science, Univ.~of Udine, 2013.

\bibitem{mp:memo14}
M.~Miculan and M.~Peressotti.
\newblock A {CSP} implementation of the bigraph embedding problem.
\newblock In T.~T. Hildebrandt, editor, {\em Proc.~MeMo}, 2014.
\newblock To appear.

\bibitem{milner:bigraphbook}
R.~Milner.
\newblock {\em The Space and Motion of Communicating Agents}.
\newblock Cambridge University Press, 2009.

\bibitem{pkss:bigactors}
E.~Pereira, C.~M. Kirsch, J.~B. de~Sousa, and R.~Sengupta.
\newblock {BigActors}: a model for structure-aware computation.
\newblock In C.~Lu, P.~R. Kumar, and R.~Stoleru, editors, {\em ICCPS}, pages
  199--208. ACM, 2013.

\bibitem{perrone:thesis}
G.~Perrone.
\newblock {\em Domain-Specific Modelling Languages in Bigraphs}.
\newblock PhD thesis, IT University of Copenhagen, 2013.

\bibitem{pdh:refine11}
G.~Perrone, S.~Debois, and T.~T. Hildebrandt.
\newblock Bigraphical refinement.
\newblock In J.~Derrick, E.~A. Boiten, and S.~Reeves, editors, {\em
  Proc.~REFINE}, volume~55 of {\em Electronic Proceedings in Theoretical
  Computer Science}, pages 20--36, 2011.

\bibitem{pdh:sac12}
G.~Perrone, S.~Debois, and T.~T. Hildebrandt.
\newblock A model checker for bigraphs.
\newblock In S.~Ossowski and P.~Lecca, editors, {\em Proc.~SAC}, pages
  1320--1325. ACM, 2012.

\bibitem{graphtransformation}
G.~Rozenberg, editor.
\newblock {\em Handbook of graph grammars and computing by graph
  transformation}, volume~1.
\newblock World Scientific, River Edge, NJ, USA, 1997.

\bibitem{sp:memo14}
M.~Sevegnani and E.~Pereira.
\newblock Towards a bigraphical encoding of actors.
\newblock In T.~T. Hildebrandt, editor, {\em Proc.~MeMo}, 2014.
\newblock To appear.

\bibitem{sevegnani2010sat}
M.~Sevegnani, C.~Unsworth, and M.~Calder.
\newblock A {SAT} based algorithm for the matching problem in bigraphs with
  sharing.
\newblock Technical Report TR-2010-311, Department of Computer Science,
  University of Glasgow, 2010.

\end{thebibliography}
}
\end{document}